\documentclass[conference]{IEEEtran} 
\usepackage{graphicx}
\usepackage[cmex10]{amsmath}
 \usepackage{amsthm} 
\usepackage{amsfonts}
\usepackage{amssymb} 
\usepackage{bbold}
\usepackage{graphicx}%
\usepackage{color}%
\usepackage{algorithmic}
 \usepackage{algorithm}
\newtheorem{theorem}{Theorem}

\usepackage[english]{babel}
\usepackage{blindtext}
\usepackage{epsfig}
\usepackage{framed}
\usepackage{xcolor}

\IEEEoverridecommandlockouts

\title{Towards a Constant-Gap Sum-Capacity Result for the Gaussian Wiretap Channel with a Helper }

\author{
\IEEEauthorblockN{Rick Fritschek}
\IEEEauthorblockA{Heisenberg Communications and Information Theory Group\\
    Freie Universit\"at Berlin, \\
    Takustr. 9,
    D--14195 Berlin, Germany\\
    Email: rick.fritschek@fu-berlin.de
 \thanks{This work was carried out as part of DFG grant WU 598/8-1 within the DFG priority program SPP 1798 (CoSIP)}.
}%
\and
\IEEEauthorblockN{Gerhard Wunder}
\IEEEauthorblockA{Heisenberg Communications and Information Theory Group\\
	Freie Universit\"at Berlin, \\
    Takustr. 9,
    D--14195 Berlin, Germany\\
Email: wunder@zedat.fu-berlin.de}

}

\begin{document}

\maketitle
\begin{abstract}
Recent investigations have shown that the sum secure degrees of freedom of the Gaussian wiretap channel with a helper is $\tfrac{1}{2}$. The achievable scheme for this result is based on the real interference alignment approach. While providing a good way to show degrees of freedom results, this technique has the disadvantage of relying on the Khintchine-Groshev theorem and is therefore limited to {\it almost all channel gains}. This means that there are infinitely many channel gains, where the scheme fails. Furthermore, the real interference alignment approach cannot be used to yield stronger constant-gap results. We approach this topic from a signal-scale alignment perspective and use the linear deterministic model as a first approximation. Here we can show a constant-gap sum capacity for certain channel gain parameters. We transfer these results to the Gaussian model and discuss the results.
\end{abstract}

\section{Introduction}

The wiretap channel, introduced by Wyner \cite{Wyner75}, represents a channel model where a user wants to communicate a message to its legitimate receiver,
without leaking information to an eavesdropper. The problem was solved for the degraded wiretap channel by Wyner, and subsequently generalized to the general wiretap channel by Csiszar and K\"orner \cite{CsizarKoerner}. Moreover, the extension towards the Gaussian case was investigated by Leung-Yan-Cheon and Hellman \cite{Hellman}. This model served a long time as a standard model for physical-layer security. However, in recent years cellular communication became increasingly important. In \cite{TekinYenerGMAC-WT} the so-called Gaussian wiretap multiple-access channel (GMAC-WT) was introduced, where the general wiretap setting is investigated in a multi-user structure. Unfortunately, a general solution was out of reach and a branch of research focused on the secure degrees of freedom (s.d.o.f) instead. The s.d.o.f are a way to measure how a systems secrecy capacity scales, with power going to infinity, in comparison to the capacity of the single-link AWGN channel. A key technique for achieving degrees-of-freedom in multi-user wireless network model without secrecy constraints is interference alignment (IA), introduced in \cite{CadambeJafarIA} and \cite{Maddah-AliKhandi-IA}, among others. The main idea of IA is to design signals of multiple users such that the interference will be aligned in a small dimension at the unintended receivers. In that way, interference-free dimensions can be maximized for the intended signals. IA methods can be divided into two categories \cite{Niesen-Ali}, the vector-space alignment approach and the signal-scale alignment approach. In the vector-space alignment approach, standard signaling dimensions such as time, frequency and multiple-antennas will be used for the alignment. The drawback of these methods are, that the channel model needs to have sufficient independent dimensions for the alignment. Single-antenna, frequency-flat and time constant channel models cannot utilize these techniques. For these channel models, the signal-scale alignment approach can be used, which utilizes signal scale dimensions. An example for such methods are lattice codes, which can be used to specifically design the used signal-scale. Signal-scale alignment techniques can be further subdivided into signal-strength deterministic models and real interference alignment. The latter one uses integer lattice transmit constellations, which are scaled such that alignment can be achieved. The intended messages are then recovered by minimum-distance decoding, where the Khintchine-Groshev theorem of Diophantine approximation theory is used to bound the minimum distance and therefore the error. As a result, decoding can be shown to work for {\it for almost all} channel gains and a number of d.o.f achievability results utilized this approach. This is an unsatisfying situation, since it still leaves an infinite amount of channel gains, for which the technique does not work. Furthermore, it leaves the impression that secrecy is dependent on the precision of the channel gain measurements. Moreover, this method is limited to d.o.f. results, meaning that it cannot be used for stronger constant-gap capacity results. On the other hand, the signal-strength deterministic approach is based on the deterministic approximation of the channel models. An example for such an approximation is the so-called linear deterministic model (LDM), introduced by \cite{Avestimehr2007}. The model is based on a binary expansion of the real signals, where the additive noise is introduced as a truncation of the resulting bit-vectors. Furthermore, superposition is modeled as binary addition on each individual bit-level and channel gains are introduced as down-shifts of the bit-vectors. As a result one gets a deterministic model, which is entirely based on the channel gains. It could be shown in \cite{Suvarup2011}, \cite{Bresler2008}, \cite{Bresler2010}, \cite{FW14b}, that a layered lattice coding approach can transfer the results of the LDM to the Gaussian models, enabling constant bit-gap results. 

{\bf Previous work and Contributions:} 
Previous work on the wiretap channel in multi-user settings mainly utilized the real IA approach in addition to cooperative jamming, introduced in \cite{TekinYenerCoopJam}. The idea of using IA in a secrecy context is to cooperatively jam the eavesdropper, while aligning the jamming signal in a small subspace at the legitimate receiver. A specialized model is the wiretap channel with a helper. This model consists of the standard wiretap channel model, with a second independent user, whose only purpose is to jam the eavesdropper. In \cite{XieUlukusWiretap-Helper} and \cite{XieUlukusWiretap-Helper2}, the real IA approach was used on the wiretap channel with a helper (with and without CSIT, respectively) to investigate the s.d.o.f, therefore achieve results for the infinite SNR regime. 
In this paper, we approach the problem of the wiretap channel with a helper from a different perspective. We follow the signal-strength approach and use the LDM to approximate the channel model and derive a constant bit-gap result for it. Furthermore, we use a layered lattice coding strategy to transform the achievable scheme to the Gaussian channel model, similar to the technique used in \cite{FW14b}. We also use previous results of \cite{Bresler2008}, to relate the converse of the LDM to the Gaussian model and thereby providing a constant bit-gap results for a certain range of channel gains. Due to the approximate nature of the LDM, the achievable scheme fails to reach the upper bound at some channel gain configurations. We analyse these cases and discuss the results.

{\bf Notation: }
We denote vectors and matrices by lower case bold and upper case bold characters, respectively. For two vectors $\mathbf{a}$ and $\mathbf{b}$, we denote by $[\mathbf{a};\mathbf{b}]$ the vector that is obtained by stacking $\mathbf{a}$ over $\mathbf{b}$. 
To specify a particular range of elements in a bit level vector we use the notation $\mathbf{a}_{[i:j]}$ to indicate that $\mathbf{a}$ is restricted to the bit levels $i$ to $j$. If $i=1$, it will be omitted $\mathbf{a}_{[:j]}$, the same for $j\!=\!n$ $\mathbf{a}_{[i:]}$. Therefore $\mathbf{a}=\mathbf{a}_{[:]}$ which would correspond to no restriction at all.

\begin{figure}
\centering
\includegraphics[scale=0.65]{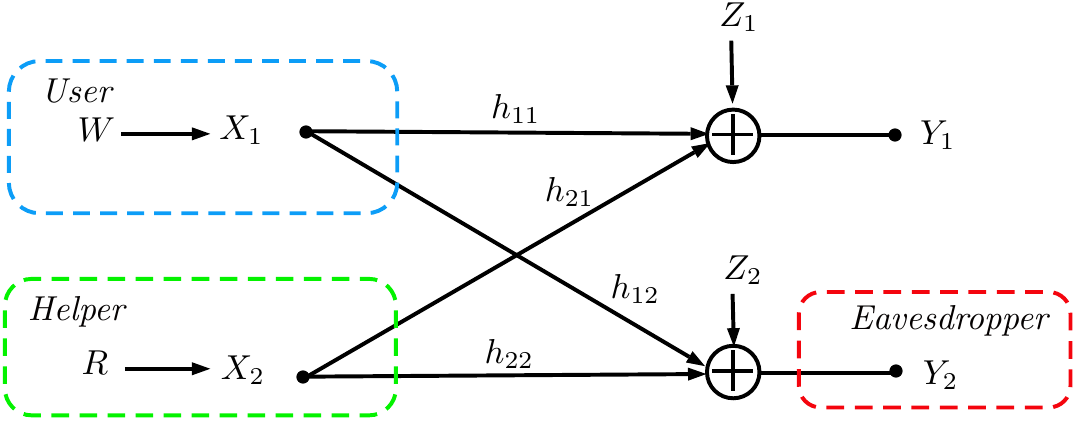}
\caption{Gaussian wiretap channel with one helper.}
\end{figure}
\label{System model}
\section{System Model} 

The Gaussian wiretap channel with a helper is defined as a system consisting of 2 transmitters and 2 receivers.
Where $X_1$ is the channel input of the user, communicating with the legitimate receiver, with channel output $Y_1$.
Whereas $X_2$ represents the channel input of an independent helper which is trying to help the user to achieve a secure transmission.
Both channel inputs are also received by an eavesdropper with channel output $Y_2$. The channel itself is modeled with additive white Gaussian noise, $Z_1,Z_2$. Therefore the system equations can be written as 
\begin{IEEEeqnarray}{rCl}\label{Gauss_Model}
Y_1&=&h_{11}X_1+h_{21}X_2+Z_1\IEEEyessubnumber\\
Y_2&=&h_{22}X_2+h_{12}X_1+Z_2\IEEEyessubnumber
\end{IEEEeqnarray},
where the channel inputs satisfy an input power constraint $E\{|X_i|^2 \} \leq P$ for each $i$. Moreover the Gaussian noise terms are assumed to be independent and zero mean with unit variance, $Z_i \sim \mathcal{CN}(0,1)$. A $(2^{nR},n)$ code will consist of an encoding and a decoding function. The encoder assigns a codeword $x_1^n(w)$ to each message $w$, where $W$ is uniformly distributed over the set $[1:2^{nR}]$, and the associated decoder assigns an estimate $\hat{w}\in [1:2^{nR}]$ to each observation of $Y_1^n$. A rate is said to be achievable if there exist a sequence of $(2^{nR},n)$ codes, for which the probability of error $P^{(n)}_e=(\hat{W}\neq W)$ goes to zero, as $n$ goes to infinity $\lim_{n\rightarrow \infty} P^{(n)}_e =0$. As opposed to the general Gaussian multiple access wiretap channel, the second channel input is used as a pure helper. This means, instead of sending codewords through $X_2$, it is used as a jamming signal. A message $W$ is said to be information-theoretically secure if the eavesdropper cannot reconstruct the message W from the channel observation $Y_2^n$. This means that the uncertainty of the message is almost equal to its entropy, given the channel observation:
\begin{equation}
\tfrac{1}{n} H(W|Y_2^n) \geq \tfrac{1}{n} H(W)-\epsilon,
\label{security_constraint}
\end{equation}
which leads to $I(W;Y_2^n)\leq \epsilon n$ for any $\epsilon>0$. A secrecy rate $r$ is said to be achievable if it is achievable while obeying the secrecy constraint \eqref{security_constraint}. 

\section{The Linear Deterministic Model System}
As simplification we will investigate the corresponding linear deterministic model (LDM) of the system model as an intermediate step. The LDM models the signals of the channel as bit-vectors $\mathbf{X}$, which is achieved by a binary expansion of the input signal $X$. The positions within the bit-vector are referred to as bit-levels. Furthermore, superposition of different signals is modeled by binary addition of the bit-levels itself. Carry over is not used to limit the superposition on the specific level where it occurs. Truncation of the bit-vector at noise level models the signal impairment of the Gaussian noise, which yields a deterministic approximation of the Gaussian model. Channel gains are included by shifting the bit-vector for an appropriate number of bit-levels. This shift is introduced by a shift-matrix $\mathbf{S}$, which is defined as
\begin{equation}
\mathbf{S}=\begin{pmatrix}
0 & 0 &  \cdots & 0 & 0\\
1 & 0 &  \cdots & 0 & 0\\
0 & 1 &  \cdots & 0 & 0\\
\vdots & \vdots & \ddots & \vdots & \vdots \\
0 & 0 &  \cdots & 1 & 0\\
\end{pmatrix}.
\end{equation}
With $\mathbf{S}$ an incoming bit vector can be shifted for $q-n$ positions with $\mathbf{Y}=\mathbf{S}^{q-n}\mathbf{X}$, where $q:=\max\{n\}$.
The channel gain is represented by $n_{ij}$-bit levels which corresponds to $\lceil\log \mbox{SNR}\rceil$ of the original channel. With this definitions, the model can be written as
\begin{IEEEeqnarray}{rCl}
\mathbf{Y}_1&=&\mathbf{S}^{q-n_{11}}\mathbf{X}'_1\oplus\mathbf{S}^{q-n_{21}}\mathbf{X}'_2\IEEEyessubnumber\\
\mathbf{Y}_2&=&\mathbf{S}^{q-n_{22}}\mathbf{X}'_2\oplus\mathbf{S}^{q-n_{12}}\mathbf{X}'_1\IEEEyessubnumber,
\label{LDM_Model}
\end{IEEEeqnarray}
where $q:=\max\{n_{11},n_{12},n_{21},n_{22}\}$. For ease of notation, we denote $\mathbf{X}_1=\mathbf{S}^{q-n_{11}}\mathbf{X}'_1$ and $\mathbf{X}_2=\mathbf{S}^{q-n_{21}}\mathbf{X}'_2$.
Furthermore, we denote $\mathbf{S}^{q-n_{22}}\mathbf{X}'_2$ and $\mathbf{S}^{q-n_{12}}\mathbf{X}'_1$ by $\bar{\mathbf{X}}_2$ and $\bar{\mathbf{X}}_1$, respectively. We assume that the helper regulates its signal such that it arrives with the same bit-level signal strength at Eve, as the signal of Alice. From a security viewpoint, this is the optimal behaviour since a weaker signal could not jam all of Alice signal and would leave bit-levels unprotected. On the other hand, if the jam signal would be stronger, it would not help to improve security either and would just increase the power load of the helper. We therefore assume that $n_{22}=n_{12}=:n_2$. 
\subsection{Achievable Scheme}
\label{ach-scheme-det-model}
\begin{figure}
\centering
\includegraphics[scale=0.9]{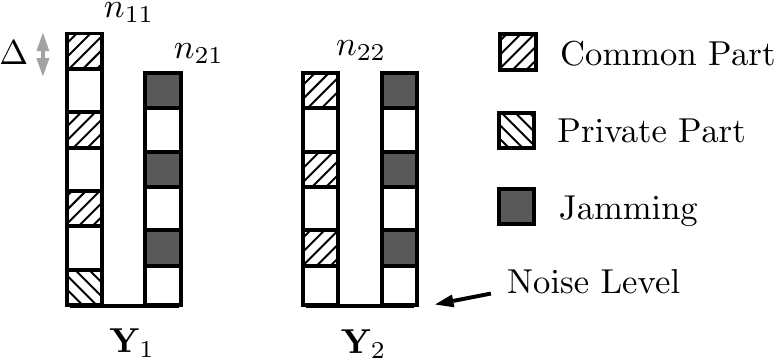}
\caption{Gaussian wiretap channel in the linear deterministic model.}
\label{wiretap-LDM}
\end{figure}

For the achievable scheme in the linear deterministic model we can make a few observations.
First of all, for $n_{11}>n_{12}$, we see that some bit-levels of the users signal fall below the noise floor at $\mathbf{Y}_2$. These bit-levels can be used for secrecy, even without a helper. We will denote this part as the private signal part. We will call the rate of the private part $r^p$. The achievability of $r^p:=(n_{11}-n_{2})^+$ is solely based on the channel properties and not on a specific scheme. With the correspondence $n \longleftrightarrow \lceil\log \mbox{SNR}  \rceil^+$, we get a rate\footnote{This rate is the maximum for cases without a helper. It is therefore an approximation of the Gaussian wiretap channel, where the capacity is known to be $C_s= \log (1+|h_{11}|^2P)- \log (1+|h_{12}|^2P)$.} of $r^p=\lceil\log |h_{11}|^2 P \rceil^+-\lceil\log |h_{12}|^2P  \rceil^+$. 
 We therefore get a partition of the secrecy rate $r:=r^p+r^c$, where $r^c$ is the achievable secrecy rate for the common part. Since the following scheme depends on the shift between the bit-vectors at Bob and at Eve, there is a singularity at $n_{11}=n_{21}$ at which the scheme fails. This is because the scheme needs a difference in channel strength in order to align the signals.

\begin{theorem}
\label{theorem_det}
The achievable secrecy rate of the linear deterministic Wiretap channel with a helper is
\begin{equation}
r_{\text{ach}}=\begin{cases}
\max\{n_{11}-n_{21},n_{21},r^p\} & \text{for } \tfrac{n_{21}}{n_{11}}<\tfrac{2}{3}
\\
r^p+r^c & \text{for } \tfrac{2}{3} \leq \tfrac{n_{21}}{n_{11}} < 2\\
n_{11} & \text{for } \tfrac{n_{21}}{n_{11}} \geq 2
\end{cases}
\end{equation}
where $r^c:=\begin{cases}
\phi_1(n^c_{Y_1},\Delta) & \text{for } n_{11}>n_{2} \text{ and } n_{11}>n_{21}
\\
\phi_2(n^{c}_{Y_1},\Delta) & \text{for everything else},
\end{cases}$
with $n^{c}_{Y_1}:=n_{11}-\,(n_{11}-n_{2})^+$ and the function $\phi_i$ for $p$,$q\in \mathbb{N}$ defined as 
\begin{equation}
\phi_1(p,q):=
\begin{cases}
\frac{(l(p,q))q}{2} & \text{if } l(p,q)\text{ is even,}
\\
p-\frac{(l(p,q)+1)q}{2} & \text{if } l(p,q)\text{ is odd},
\end{cases}
\label{phi}
\end{equation} and
\begin{equation}
\phi_2(p,q):=
\begin{cases}
\frac{(l(p,q)+1)q}{2} & \text{if } l(p,q)\text{ is odd,}
\\
p-\frac{l(p,q)q}{2} & \text{if } l(p,q)\text{ is even},
\end{cases}
\label{phi}
\end{equation}
with $\Delta:=|n_{11}-n_{21}|$, $l(p,q):=\lfloor\frac{p}{q}\rfloor\ \mbox{for}\ q>0\ \mbox{and}\ l(p,0)=0$.

\end{theorem}
\begin{proof} In this part we will prove the achievability of the common secrecy rate $r^c$. The achievable scheme relies on a partition of the bit-levels into intervals. Every bit-vector gets partitioned into $\Delta:=|n_{11}-n_{21}|$-bit level sized partitions, starting at the top, i.e. the most significant bit. For example, the user signal at $\mathbf{Y}_1$ consists of $\lfloor\frac{n_{11}}{\Delta}\rfloor$ $\Delta$ sized partitions and a remainder partition with $n_{11}-\lfloor\frac{n_{11}}{\Delta}\rfloor$ bit-levels. The partitions of the bit-vectors will be enumerated from the top downwards. Now, all odd partitions of $\mathbf{X}_1$ will be used for encoding of the message, while all odd partitions of $\bar{\mathbf{X}}_2$ will be used for jamming. Note that we used the jamming signal at $\mathbf{Y}_2$ as basis, because the need for jamming is limited to the wiretapper. Which means that in case $n_{21}>n_{2}$, the part $\mathbf{X}_{2,{[n_{2}+1:]}}$ can be left unused without secrecy penalty. Due to the $\Delta$-partitioning, the $\Delta$-shift will result in an exact alignment of the message partitions with the jamming partitions at $\mathbf{Y}_2$. Moreover, at $\mathbf{Y}_1$, the used partitions will align with the unused partitions, enabling the secrecy rate. We can therefore sum over the bit-allocation of $\mathbf{X}_1$ to get the common part secrecy rate. For an odd number of partitions, the remainder partition is even (and therefore unused) and we just have to sum over $\lfloor\frac{n_{11}}{\Delta}\rfloor$ partitions. Since only every second partition is used, and accounting for an odd total number of partitions, the rate $r^c$ follows for the case $l(n^{c}_{Y_1} ,\Delta)$ odd. Similarly for the case of an even number of partitions. In that case one needs to account for the remainder partition, which will be odd and is therefore also used. As a special case, we need to consider the situation with a non-zero private part, and weaker helper signal at $\mathbf{Y}_1$. In this case, we cannot utilize the remainder-part, since the jamming signal of the helper would fall in the private part range. This results in one less allocatable $\Delta$-partition and yields the function $\phi_1$. For the range $ \tfrac{n_{21}}{n_{11}}<\tfrac{2}{3}$, one can allocate the whole helper-free common bit-levels at $\mathbf{Y}_1^n$. Together with the private part, the achievable rate is reached. For the range of $\tfrac{n_{21}}{n_{11}} \geq 2$, the helper signal is strong enough to jam the whole user signal at the eavesdropper, without affecting the transmission to the legitimate receiver. Therefore the top $n_{11}$ bits of $\mathbf{X}_2$ are used for jamming, while all of $\mathbf{X}_1$ is used for signal transmission.
\end{proof}

\subsection{Converse}
\label{det-converse}
\begin{theorem}
The secrecy rate $R$ of the linear deterministic wiretap channel with one helper and symmetric channel gains at the wiretapper is bounded from above by

\begin{equation*}
R \leq \min\{ r_{ub1}, r_{ub2}, r_{ub3}\}
\end{equation*}
with \begin{IEEEeqnarray*}{rCl}
r_{ub1}&=& (n_{11}-n_{2})^++\tfrac{1}{2}(\max\{n_{11},n_{21}\}-(n_{11}-n_{2})^+)\\
&&+\:\tfrac{1}{2} (n_{2}-n_{21})^+\\
r_{ub2}&=&n_{11}\\
r_{ub3}&=&n_{21}+(n_{11}-n_{21}-n_{2})^+\\
&& +\:[n_2-n_{21}-(n_2-n_{11}+n_{21})^+]^+
\end{IEEEeqnarray*}
\end{theorem}

\begin{proof}

We start as in \cite{XieUlukusWiretap-Helper} with the following procedure
\begin{IEEEeqnarray*}{rCl}
nR&=& H(W|\mathbf{Y}_1^n)+I(W;\mathbf{Y}_1^n)\\
&\leq & I(W;\mathbf{Y}_1^n)+ n\epsilon\\
&\leq & I(W;\mathbf{Y}_1^n)-I(W;\mathbf{Y}_2^n)+n\epsilon_2\\
&= & H(\mathbf{Y}_1^n)-H(\mathbf{Y}_1^n|W)-H(\mathbf{Y}_2^n)+H(\mathbf{Y}_2^n|W)+n\epsilon_2\\
&=& H(\mathbf{Y}_1^n)-H(\mathbf{Y}_2^n)+H(\bar{\mathbf{X}}_2^n)-H(\mathbf{X}_2^n)+n\epsilon_2\IEEEyesnumber
\label{startEQN}
\end{IEEEeqnarray*} where Fanos inequality and the secrecy constraint was used. In the last line we used that $X_1$ is a function of $W$, and $X_2$ is independent of $W$. We remark that the first property does not hold in general, since jamming through the first user would result in a stochastic function. For the following analyses we have to distinguish between the three cases $n_{21}>n_{2}$, $n_{21}< n_{2}$ and $n_{21}= n_{2}$.
We begin with the entropy difference $H(\mathbf{Y}_1^n)-H(\mathbf{Y}_2^n)$. First of all, in case that $n_{21}< n_{2}$, we have to split up the negative entropy part in
$H(\mathbf{Y}_2^n)=H(\mathbf{Y}_{2,[:n_{21}]}^n)+H(\mathbf{Y}_{2,[n_{21}+1:]}^n|\mathbf{Y}_{2,[:n_{21}]}^n)$. The latter term will be used with the terms $H(\bar{\mathbf{X}}_2^n)-H(\mathbf{X}_2^n)$ which yields 
\begin{IEEEeqnarray*}{rCl}
&&H(\bar{\mathbf{X}}_2^n)-H(\mathbf{X}_2^n)-H(\mathbf{Y}_{2,[n_{21}+1:]}^n|\mathbf{Y}_{2,[:n_{21}]}^n)\\
&\leq & H(\bar{\mathbf{X}}_2^n)-H(\mathbf{X}_2^n)-H(\mathbf{Y}_{2,[n_{21}+1:]}^n|\mathbf{Y}_{2,[:n_{21}]}^n,\mathbf{X}_1^n)\\
&=&0
\end{IEEEeqnarray*}

In the following proof section we assume that $n_{21}>n_{2}$. The case $n_{21}< n_{2}$ follows in a similar matter.
We define $\eta:=\max\{n_{11},n_{21}\}-(n_{11}-n_{2})^+$ and $\eta':=n_{21}-(n_{11}-n_{2})^+$ to split
the received signals in common and private parts. We start by adding two of the terms and show
\begin{IEEEeqnarray*}{rCl}
&&2(H(\mathbf{X}_1^n\oplus \mathbf{X}_2^n)-H(\bar{\mathbf{X}}_1^n\oplus \bar{\mathbf{X}}_2^n)) \\
&\leq & 2H((\mathbf{X}_1^n\oplus \mathbf{X}_2^n)_{[\eta+1:]})+ 2H((\mathbf{X}_1^n\oplus \mathbf{X}_2^n)_{[:\eta]})\\
&&\:-2H(\bar{\mathbf{X}}_1^n\oplus \bar{\mathbf{X}}_2^n).
\end{IEEEeqnarray*}
Note that the private part $H((\mathbf{X}_1^n\oplus \mathbf{X}_2^n)_{[\eta+1:]})$ is zero for $n_{11}\leq n_{2}$.
In the next step, we will analyse the entropy difference. We will use a method inspired by \cite{Fritschek2014} to show the following
\begin{IEEEeqnarray*}{rCl}
&&2(H(\mathbf{X}_1^n\oplus \mathbf{X}_2^n)-H(\bar{\mathbf{X}}_1^n\oplus \bar{\mathbf{X}}_2^n)) \\
&\leq & 2H((\mathbf{X}_1^n\oplus \mathbf{X}_2^n)_{[\eta+1:]})+ 2H((\mathbf{X}_1^n\oplus \mathbf{X}_2^n)_{[1:\eta]})\\
&&-\:H(\bar{\mathbf{X}}_1^n\oplus \bar{\mathbf{X}}_2^n|\bar{\mathbf{X}}_1^n)-H(\bar{\mathbf{X}}_1^n\oplus \bar{\mathbf{X}}_2^n|\bar{\mathbf{X}}_2^n)\\
&= & 2H((\mathbf{X}_1^n\oplus \mathbf{X}_2^n)_{[\eta+1:]})+ 2H((\mathbf{X}_1^n\oplus \mathbf{X}_2^n)_{[1:\eta]}|)\\
&&\:-H(\bar{\mathbf{X}}_{2}^n)-H(\bar{\mathbf{X}}_{1}^n)\\
&\leq & 2H((\mathbf{X}_1^n\oplus \mathbf{X}_2^n)_{[\eta+1:]})+ H((\mathbf{X}_1^n\oplus \mathbf{X}_2^n)_{[1:\eta]})\\
&&\:+H( \mathbf{X}_{2,[1:\eta']}^n)+ H(\mathbf{X}_{1,[1:\eta]}^n)-H(\bar{\mathbf{X}}_{2}^n)-H(\bar{\mathbf{X}}_{1}^n).\\
\end{IEEEeqnarray*}

We now see that $H(\mathbf{X}_{1,[1:\eta'']}^n)-H(\bar{\mathbf{X}}_{1}^n)\leq 0$, because the first term is just the common part and we have that
$H( \mathbf{X}_{2,[1:\eta']}^n)-H(\bar{\mathbf{X}}_{2}^n)$ depends on the actual channel gains. For $n_{21}< n_{2}$ these terms have the same strength due to the split of $H(\mathbf{Y}_2^n)$ resulting in zero. For $n_{21}>n_{2}$ we get 
$H(\mathbf{X}_{2,[n_{2}+1:\eta']}^{n}|\mathbf{X}_2^{n,c})$.
Now one can divide all terms by two, resulting in
\begin{IEEEeqnarray*}{rCl}
&&H(\mathbf{Y}_1^n)-H(\mathbf{Y}_2^n) \\
&\leq & H((\mathbf{X}_1^n\oplus \mathbf{X}_2^n)_{[\eta+1:]})+ \frac{1}{2}H((\mathbf{X}_1^n\oplus \mathbf{X}_2^n)_{[:\eta]})\\
&&\:+\frac{1}{2} (H( \mathbf{X}_{2,[1:\eta']}^n)-H(\bar{\mathbf{X}}_{2}^n)).
\end{IEEEeqnarray*}
We finalize the proof by plugging this result into \eqref{startEQN}, which results in
\begin{IEEEeqnarray*}{rCl}
nR&\leq& H(\mathbf{Y}_1^n)-H(\mathbf{Y}_2^n)+H(\bar{\mathbf{X}}_2^n)-H(\mathbf{X}_2^n)+n\epsilon_2\\
&\leq & n(n_{11}-n_{2})^++n\tfrac{1}{2}(\max\{n_{11},n_{21}\}-(n_{11}-n_{2})^+)+\\
&&+\:n\tfrac{1}{2} (n_{2}-n_{21})^++n\epsilon_2
\end{IEEEeqnarray*}
dividing by n and letting $n\rightarrow \infty$ shows the result.

For the case that $\beta_1>2$ we have that
\begin{IEEEeqnarray*}{rCl}
n(R-\epsilon_2) &\leq &  H(\mathbf{Y}_1^n)-H(\mathbf{Y}_2^n)+H(\bar{\mathbf{X}}_2^n)-H(\mathbf{X}_2^n)\\
& \leq & H(\mathbf{X}_1^n)+H(\mathbf{X}_2^n)-H(\mathbf{Y}_2^n|\bar{\mathbf{X}}_1^n)+H(\bar{\mathbf{X}}_2^n)-H(\mathbf{X}_2^n)\\
&=& H(\mathbf{X}_1^n)\\
& \leq & n n_{11}
\end{IEEEeqnarray*}
and for the case that $\beta_1<\tfrac{2}{3}$ we have that
\begin{IEEEeqnarray*}{rCl}
nR &\leq &  H(\mathbf{Y}_1^n)-H(\mathbf{Y}_2^n)+H(\bar{\mathbf{X}}_2^n)-H(\mathbf{X}_2^n)+n\epsilon_2\\
& \leq & H(\mathbf{Y}_{1,[:(n_{11}-n_{21})]}^n)+ H(\mathbf{Y}_{1,[(n_{11}-n_{21})+1:]}^n|\mathbf{Y}_{1,[:(n_{11}-n_{21})]}^n)\\
&&-\: H(\mathbf{Y}_2^n)+H(\bar{\mathbf{X}}_2^n)-H(\mathbf{X}_2^n)+n\epsilon_2\\
& \leq & H(\mathbf{Y}_{1,[:(n_{11}-n_{21})]}^n)+ H(\mathbf{Y}_{1,[(n_{11}-n_{21})+1:]}^n|\mathbf{Y}_{1,[:(n_{11}-n_{21})]}^n)\\
&&-\: H(\mathbf{Y}_{2,[:(n_{11}-n_{21})]}^n)-H(\mathbf{Y}_{2,[(n_{11}-n_{21})+1:]}^n|\mathbf{Y}_{2,[:(n_{11}-n_{21})]}^n,\mathbf{X}_1^n)\\
&&+\: H(\bar{\mathbf{X}}_2^n)-H(\mathbf{X}_2^n)+n\epsilon_2.
\end{IEEEeqnarray*}

One can show that \begin{IEEEeqnarray*}{rCl}
&&H(\mathbf{Y}_{1,[:(n_{11}-n_{21})]}^n)-H(\mathbf{Y}_{2,[:(n_{11}-n_{21})]}^n) \\
&\leq & n(n_{11}-n_{21}-n_{2})^+
\end{IEEEeqnarray*}
and \begin{IEEEeqnarray*}{rCl}
&&H(\bar{\mathbf{X}}_2^n)-H(\mathbf{Y}_{2,[(n_{11}-n_{21})+1:]}^n|\mathbf{Y}_{2,[:(n_{11}-n_{21})]}^n,\mathbf{X}_1^n)-H(\mathbf{X}_2^n)\\
&\leq & n[n_2-n_{21}-(n_2-n_{11}+n_{21})^+]^+
\end{IEEEeqnarray*}
and $H(\mathbf{Y}_{1,[(n_{11}-n_{21})+1:]}^n|\mathbf{Y}_{1,[:(n_{11}-n_{21})]}^n) \leq nn_{21}$ which yields
\begin{IEEEeqnarray*}{rCl}
nR & \leq & H(\mathbf{Y}_{1,[:(n_{11}-n_{21})]}^n)+ H(\mathbf{Y}_{1,[(n_{11}-n_{21})+1:]}^n|\mathbf{Y}_{1,[:(n_{11}-n_{21})]}^n)\\
&&-\: H(\mathbf{Y}_{2,[:(n_{11}-n_{21})]}^n)\\
&&-\:H(\mathbf{Y}_{2,[(n_{11}-n_{21})+1:]}^n|\mathbf{Y}_{2,[:(n_{11}-n_{21})]}^n,\mathbf{X}_1^n)\\
&&+\: H(\bar{\mathbf{X}}_2^n)-H(\mathbf{X}_2^n)+n\epsilon_2\\
&\leq & nn_{21}+n(n_{11}-n_{21}-n_{2})^+\\
&& +\:n[n_2-n_{21}-(n_2-n_{11}+n_{21})^+]^++n\epsilon_2
\end{IEEEeqnarray*}

dividing by n and letting $n\rightarrow \infty$ shows the result.

\end{proof}

\section{The Gaussian wiretap channel with a helper}

In this section we analyse the Gaussian wiretap channel with a helper. To get results we stick to the previously developed scheme in section \ref{ach-scheme-det-model}, and we will transfer the alignment and jamming structure to its Gaussian equivalent with layered lattice codes. This will lead to an achievable rate which is directly based on the deterministic rate. Moreover, we will make use of results in \cite{Bresler2008} to show that the mutual information of the Gaussian case can be upper bounded by the deterministic terms. As a result, the deterministic bound in section \ref{det-converse} is a bound for the Gaussian model as well, with a constant bit-gap attached. 

\subsection{Achievable Scheme}

\begin{theorem}
The achievable secrecy rate of the Gaussian Wiretap channel with a helper is
\begin{equation}
r_{\text{ach}}=\begin{cases}
\max\{\log(\text{SNR}_1^{(1-\beta_1)}),\log(\text{SNR}_2),r^p\}  & \text{for } \beta_1<\tfrac{2}{3}
\\
r^p+r^c-d & \text{for } \tfrac{2}{3} \leq \beta_1 < 2\\
\log(\text{SNR}_1) & \text{for } \beta_1\geq 2
\end{cases}
\end{equation}
where $r^p:=\log(\text{SNR}_1^{1-\beta_2})$
and
$r^c:=\\\begin{cases}
\phi_1(\log \text{SNR}_1^{(1-(1-\beta_2)^+)},\log \text{SNR}_1^{(1-\beta_1)}) & \text{for }\beta_1, \beta_2 \leq 1,
\\
\phi_2(\log \text{SNR}_1^{(1-(1-\beta_2)^+)},\log \text{SNR}_1^{(1-\beta_1)}) & \text{for everything else},
\end{cases}$
with the function $\phi_i$ as defined in \eqref{theorem_det} and $d=\left\lfloor l_{max} \right\rfloor$.
\end{theorem}
\begin{proof}
For the achievable scheme, we need to partition the available power into intervals. Each of these intervals plays the role of an $\Delta-$Interval of bit-levels in the linear deterministic scheme. Remember that we have $E\{|X_i|^2\}\leq P$ and $Z_1,Z_2 \sim \mathcal{CN}(0,1)$, which means that $|h_{11}|^2P=\text{SNR}_1$ and $|h_{21}|^2P=\text{SNR}_2$ represent the power of both direct signals. As in the deterministic model, we assume that both signals at $Y_2$ are received with the same power and therefore $h_{12}=h_{22}=h_2$ which gives $|h_2|^2P=\text{SNR}_3$. We introduce the two parameters $\beta_1$ and $\beta_2$, which connects the $\text{SNR}$ ratios with $\text{SNR}_2=\text{SNR}_1^{\beta_1}$ and $\text{SNR}_3=\text{SNR}_1^{\beta_2}$. Now we can partition the received power at $Y_1$ into intervals of $\text{SNR}_1^{(1-\beta_1)}$. Each of the intervals has therefore signal power $\theta_l$ which is defined as 
\begin{IEEEeqnarray}{rCl}
\theta_l & = & q_{l-1}-q_l\IEEEnonumber\\
&= & \text{SNR}_1^{1-(l-1)(1-\beta_1)}-\text{SNR}_1^{1-l(1-\beta_1)}
\label{theta_k_fractions}
\end{IEEEeqnarray}
with $l$ indicating the specific level. The users decompose the signals $X_i$ into a sum of independent sub-signals $X_i=\sum_{l=1}^{l_{max}}X_{il}$. For each of these sub-signals, a lattice code is chosen as described in \cite{Loeliger}, such that it is good for channel coding with an average power per dimension of $\theta_l$. $X_1$ is used for signal transmission, while $X_2$ is solely used for jamming. As in the deterministic case, the objective is to align the signal parts of $X_1$ with the jamming of $X_2$ at $Y_2$, while allowing decoding of the signal parts at $Y_1$. Due to the signal scale based coding strategy and the equal receive power at $Y_2$, an alignment is achieved with the proposed scheme. We therefore need to prove, that the signal can be decoded at $Y_1$.
\subsubsection{Decoding procedure}
The decoding is done level-wise, treating subsequent levels as noise. Every level is treated as a Gaussian point-to-point channel with power $\theta_l$ and noise $1+2\text{SNR}_1^{1-l(1-\beta_1)}$, which consists of the base noise $N_1$ at $Y_1$ and the power of all subsequent levels of both signals. Successful decoding can be assured with a rate limitation of 
\begin{equation}
r_l\leq \log \left(\frac{\theta_l}{1+2\text{SNR}_1^{1-l(1-\beta_1)}}\right),
\label{Decodingbound}
\end{equation} which ensures that a lattice code $(\gamma\Lambda_C+v) \cap S$ exists with arbitrary small error probability \cite{Loeliger}. This code consists of a lattice $\Lambda_C \in \mathbb{R}^n$, a scaling factor $\gamma \in \mathbb{R}$, a translation $v\in \mathbb{R}^n$ and a spherical shaping region $S\subset\mathbb{R}^n$ with power $\theta_l$ per dimension.
\subsubsection{Achievable rate}
As in the deterministic case, we have a private part in case that $\beta_2 < 1$, which can be used completely. Since it has only the base noise, a rate of $r^p=\log(\text{SNR}_1^{(1-\beta_2)})=\log(\text{SNR}_1)-\log(\text{SNR}_3)$ can be achieved. For the common part we need to look into the scheme itself. All odd levels of $X_1$ will be used for signal transmission. Every level $l$ can handle a rate of $r_l$. We have a total of $\lceil l_{max}\rceil$ levels in $X_1$, where $l_{max}:=\frac{1}{1-\beta_1}$. This sums up to a rate of
\begin{equation}
r^c=\sum_{l=1}^{\left\lfloor l_{max}\right\rfloor} \mathbb{1}_{odd} \log \left( \frac{\theta_l}{1+2\text{SNR}_1^{1-l(1-\beta_1)}} \right),
\label{Gauss_rate}
\end{equation}
where we count only the odd levels towards the achievable rate. Moreover, we need to consider the remainder term, which is allocated between the alignment structure and the noise floor or the private part (if $\beta_3<1$). This term is zero for an odd number of full levels and it is used, if there is an even number of full levels. In that case the remainder term would yield a rate of $r^R\leq \log (\text{SNR}_1^{1-\lceil l_{max} \rceil(1-\beta_1)}-1)$. Furthermore, we can simplify the rate of \eqref{Gauss_rate} with
\begin{IEEEeqnarray*}{rCl}
\log \left(\frac{a-b}{1+cb} \right)&> & \log \left(\frac{a-b}{1+b} \right)-\log c\\
&>& \log \left(\frac{1+a}{1+b} \right)-\log (c)-1\\
&>& \log \left(\frac{a}{2b} \right) -\log (c)-1
\end{IEEEeqnarray*} where we used that $c,b>1$ and $\frac{a-b}{1+b}>1$ and get 
\begin{IEEEeqnarray*}{rCl}
r^c&>&\sum_{l=1}^{\left\lfloor l_{max}\right\rfloor} \mathbb{1}_{odd} \log \left( \frac{\text{SNR}_1^{1-(l-1)(1-\beta_1)}}{2\text{SNR}_1^{1-l(1-\beta_1)}} \right) - 2\\
&>& \tfrac{\left\lfloor l_{max}\right\rfloor}{2} \log \text{SNR}_1^{(1-\beta_1)} - \left\lfloor l_{max} \right\rfloor
\end{IEEEeqnarray*}
Combined with the remainder part, we see that we get $\phi_1(\log \text{SNR}_1^{(1-(1-\beta_2)^+)}, \log \text{SNR}_1^{(1-\beta_1)})$ and for the case of $\beta_2<1$ we need to handle the private part. After some modification, one can see that the same technique yields a common rate of $\phi_2(\log \text{SNR}_1^{(1-(1-\beta_2)^+)}, \log \text{SNR}_1^{(1-\beta_1)})$. 

\end{proof}

\subsection{Converse}

For the converse, we make use of the linear deterministic bound. We start with 
$ nR=I(W;Y_1^n) - I(W;Y_2^n) + n\epsilon$, where we need to bound the two mutual information terms. We can use a result of \cite[Thm.1]{Bresler2008} for the complex Gaussian IC, which shows that the capacity is within 42 bits of the deterministic IC capacity. 
We can then re-use the deterministic techniques, to show the converse, leading to the following theorem.
\begin{theorem}
The secrecy rate $R$ of the Gaussian wiretap channel with one helper and symmetric channel gains at the wiretapper is bounded from above by

\begin{equation*}
R \leq \min\{ r_{Gub1}, r_{Gub2}, r_{Gub3}\}
\end{equation*}
with \begin{IEEEeqnarray*}{rCl}
r_{Gub1}&=& (n_{11}-n_{2})^++\tfrac{1}{2}(\max\{n_{11},n_{21}\}-(n_{11}-n_{2})^+)\\
&&+\:\tfrac{1}{2} (n_{2}-n_{21})^++c\\
r_{Gub2}&=&n_{11}+c\\
r_{Gub3}&=&n_{21}+(n_{11}-n_{21}-n_{2})^+\\
&& +\:[n_2-n_{21}-(n_2-n_{11}+n_{21})^+]^++c
\end{IEEEeqnarray*}
where $c$ is a constant.
\end{theorem}
\begin{proof}
\begin{IEEEeqnarray*}{rCl}
nR&\leq & I(W;Y_1^n)-I(W;Y_2^n)+n\epsilon\\
&\leq & I(W;\mathbf{Y}_1^n)-I(W;\mathbf{Y}_2^n)+n\epsilon +nc\\
&\leq&  n(n_{11}-n_{2})^++\tfrac{n}{2} (n_{2}-n_{21})^+\\
&&+\:\tfrac{n}{2}(\max\{n_{11},n_{21}\}-(n_{11}-n_{2})^+)+n\epsilon+nc.
\end{IEEEeqnarray*}
Dividing by $n$ and taking $n\rightarrow \infty$ shows the result of the first bound. The  other bounds can be shown similarly.
\end{proof}

\section{Conclusions}
We have shown an achievable scheme and a converse bound for the sum capacity of the Gaussian wiretap channel with a helper which is tight for a certain range of parameters (see Fig. \ref{beta_1}). We used the linear deterministic approximation of the model, to gain insights into the model structure, and transferred those results to the Gaussian model, within a constant gap. These techniques can be summarized as signal-scale alignment methods, where we used jamming alignment at the eavesdropper in the signal-scale, while minimizing the negative effect on the decoding at the legitimate receiver. Looking into figures \ref{beta_2} and \ref{beta_1}, one can see the achievable rate scaled by the single-link channel, with varying parameters $\beta_1,\beta_2$, i.e. channel gain configurations. These figures can be interpreted as a form of {\it{generalized secure degrees of freedom}}, where the minimum over the whole parameter range shows the s.d.o.f of the system model. One can see that the figures show s.d.o.f of $\frac{1}{2}$, which agrees with the results of \cite{XieUlukusWiretap-Helper}. Also note that in Fig.~\ref{beta_1}, in the range of $\tfrac{2}{3}\leq\beta_1<2$, the achievable scheme fluctuates between the upper bound and $\tfrac{1}{2}$. We believe that this is a result of the approximation error, introduced by the deterministic approximation and the bit-level techniques which also gets transferred to the Gaussian model. A more sensitive deterministic model, like the one used in \cite{Niesen-Ali}, could help to completely reach the upper bound. This would give a constant-gap sum-capacity result for the whole range. 

\begin{figure}
\centering
\includegraphics[scale=0.9]{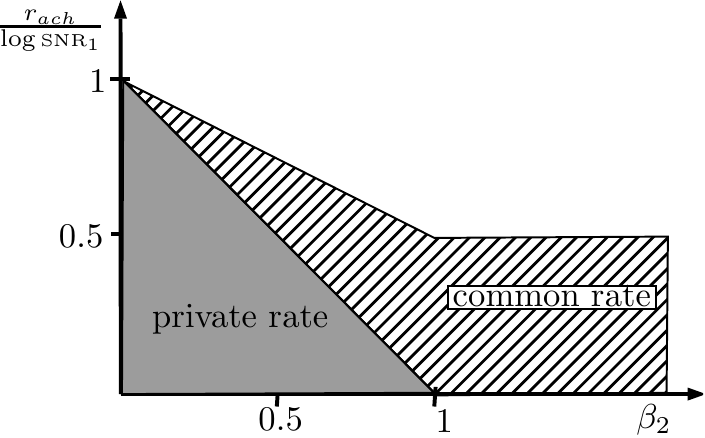}
\caption{Illustration of the achievable secrecy rate in relation to the single-link scenario, and variation in the $\beta_2$ parameter, while $\beta_1$is fixed at 0.75.}
\label{beta_2}
\end{figure}
\begin{figure}
\centering
\includegraphics[scale=0.9]{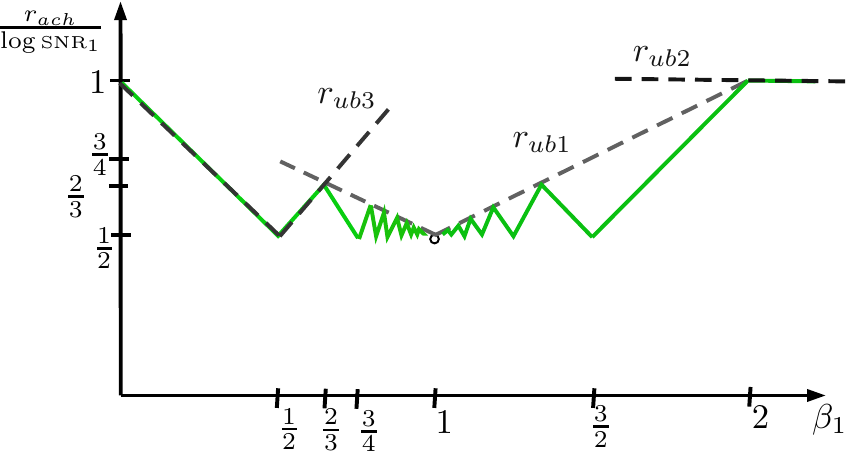}
\caption{Illustration of the achievable secrecy rate in relation to the single-link scenario, and variation in the $\beta_1$ parameter, while $\beta_2$ is fixed at 1.}
\label{beta_1}
\end{figure}

\bibliographystyle{./IEEEtran}
\bibliography{./ref}

\end{document}